\documentclass[lettersize,journal]{IEEEtran}

\usepackage{amsmath}
\usepackage{amssymb}
\usepackage{amsthm}
\usepackage{mathtools}
\usepackage{graphicx}
\usepackage{booktabs}
\usepackage{array}
\usepackage{xcolor}
\usepackage{cite}
\usepackage[colorlinks=true, allcolors=blue]{hyperref}

\newtheorem{theorem}{Theorem}
\newtheorem{corollary}{Corollary}
\theoremstyle{definition}
\newtheorem{assumption}{Assumption}
\theoremstyle{remark}
\newtheorem{remark}{Remark}

\newcommand{\E}{\mathbb{E}}
\newcommand{\R}{\mathbb{R}}
\newcommand{\C}{\mathbb{C}}
\newcommand{\w}{\mathbf{w}}
\newcommand{\g}{\mathbf{g}}
\newcommand{\avec}{\mathbf{a}}
\newcommand{\grad}{\nabla}
\newcommand{\Lcal}{\mathcal{L}}
\newcommand{\Rcal}{\mathcal{R}}
\newcommand{\norm}[1]{\left\|#1\right\|}
\newcommand{\abs}[1]{\left|#1\right|}
\newcommand{\inner}[2]{\langle #1,\, #2 \rangle}
\newcommand{\thetahat}{\hat{\theta}}
\newcommand{\wref}{\w^{(0)}}

\usepackage{xspace}
\newcommand{\feddoa}{FedDoA\xspace}
\newcommand{\doaProx}{DoAProx\xspace}

\begin{document}

\title{Physics-Informed Direction-of-Arrival Estimation Over Distributed Edge Devices}

\author{Nathan Tatsuta and Rajeev Sahay
\thanks{N.~Tatsuta and R.~Sahay are with the Department of Electrical
and Computer Engineering, University of California San Diego, San Diego,
CA, 92093 USA. E-mail: \url{ntatsuta,r2sahay@ucsd.edu}.}
\thanks{This work was supported in part by the UC San Diego Academic
Senate under grant RG116457 and in part by the National Science
Foundation (NSF) under grant ECCS-2512912.}}

\maketitle

\begin{abstract}
Direction-of-arrival (DoA) estimation is a fundamental array processing task that has benefited substantially from deep learning. Deploying such methods using data distributed across multiple edge devices introduces communications overhead, due to the aggregation of large datasets that federated learning (FL) can address. Yet, standard FL algorithms treat DoA as a generic classification problem, ignoring the underlying physics of the array manifold. To address this, we propose a physics-informed FL framework for DoA estimation that incorporates steering-vector geometry directly into the local training objective via a manifold-aware regularizer. Unlike existing FL baselines, the regularizer in our framework penalizes discrepancies in steering space rather than label space, exploiting the known geometric structure of the array manifold. We provide theoretical convergence guarantees for our framework, showing convergence to a stationary point. Simulation results confirm that our physics-informed approach outperforms multiple FL baseline approaches across iid and non-iid data conditions.
\end{abstract}

\begin{IEEEkeywords}
Array processing, direction of arrival, federated learning, physical-layer communications.
\end{IEEEkeywords}

\section{Introduction}


\IEEEPARstart{D}{irection}-of-arrival (DoA) estimation is a cornerstone problem in array processing, with applications spanning communications, radar, navigation, and localization. Recent deep learning (DL) based approaches have demonstrated robust DoA performance in comparison to classical methods (i.e., maximum likelihood), particularly under poor channel conditions \cite{liu2018direction,papageorgiou2021deep,shmuel2024subspacenet}. A key limitation of DL-based DoA, however, lies in the need to centralize large datasets from multiple receivers, incurring high communications overhead in distributed sensing systems. Federated learning (FL), a distributed DL training paradigm, addresses this challenge by locally training models on edge devices, keeping raw data local, and only transmitting model weights over the air ~\cite{mcmahan2017communication}.

Despite its advantage, however, standard FL algorithms such as federated averaging \cite{mcmahan2017communication}, FedProx \cite{li2020federated}, and SCAFFOLD \cite{karimireddy2020scaffold} optimize objectives that minimize generic cross-entropy loss functions, where all angular misclassifications are equally penalized rather than dynamically penalizing close incorrect predictions (e.g., when $\hat\theta = 6^\circ$ and $\theta = 8^\circ$) less than large incorrect predictions (e.g., when $\hat\theta = 30^\circ$ and $\theta = 6^\circ$), thus ignoring the geometry of the ULA manifold in which such information is captured. Under statistically heterogeneous data distributions (i.e., non-identical label/channel distributions across clients), where different arrays observe signals from different angular sectors or different channel conditions, this mismatch compounds, resulting in client models drifting in directions that are geometrically inconsistent with the array manifold and, ultimately, degrading the global model \cite{zhang2022federated, bhatti2026channel}. 

In this work, we propose a \emph{physics-informed} FL framework for DoA that optimizes its estimator by explicitly accounting for the steering-vector regularizer. The regularizer penalizes the Euclidean distance between the predicted and true steering vectors on the array manifold, directly encoding the metric structure of the ULA response into the training objective. Compared to prior FL work in array processing, where FL networks are generically applied at the edge, we exploit domain knowledge about the physical channel into the FL training paradigm. Moreover, our approach is network architecture-agnostic and can be integrated into existing networks. 

Our \textbf{specific contributions} are as follows.
\begin{itemize}


    \item \textbf{Physics-Informed DoA Estimation:} We propose a physics-informed FL framework for DoA that penalizes discrepancies in steering space rather than in label space. 
    
    \item \textbf{Convergence Guarantees:} We establish and empirically validate the first rigorous non-convex convergence analysis for a federated DoA method. In addition, we derive the first angular error bound for federated DoA. 
    
    \item \textbf{Complementarity with Drift:} We show that the steering regularizer and weight-space proximal term act in complementary domains rather than as interchangeable terms. 

\end{itemize}


\vspace{-0.4cm}
\section{Related Work}

Centralized DL has been widely applied to DoA estimation \cite{liu2018direction,papageorgiou2021deep,merkofer2023music,shmuel2024subspacenet,zheng2024deep}, but these approaches assume centralized training data. In distributed settings, FL has been applied to related array-processing tasks such as MIMO channel estimation \cite{elbir2022federated} and joint THz channel/DoA estimation \cite{elbir2023federated}, but no prior work addresses client heterogeneity or incorporates a physics-informed objective. Physics-informed regularizers have shown broad benefits elsewhere~\cite{karniadakis2021physics} but remain unexplored for distributed DoA. Alternatively, our framework embeds the ULA steering-manifold geometry directly into the local objective, allowing the physics of DoA estimation to guide gradient updates in a way that weight-space regularizers cannot. 


\vspace{-0.3cm}
\section{Methodology}


\subsection{Signal Model}

We consider a uniform linear array (ULA) of $M$ elements with inter-element spacing $d = \lambda/2$, where $\lambda$ is the carrier wavelength. A narrowband far-field source impinges from direction $\theta$. The received signal vector at snapshot $n$ is
\begin{equation}
  \mathbf{z}(n) = \avec(\theta)\,\mathbf{s}(n) + \mathbf{n}(n), \quad n = 1,\ldots,N,
  \label{eq:sig}
\end{equation}
where $\mathbf{s}(n)$ is the source signal (including channel fading), $\mathbf{n}(n) \sim \mathcal{CN}(\mathbf{0},\sigma^2\mathbf{I}_M)$ is AWGN, and $\avec(\theta) \in \mathbb{C}^M$ is the ULA steering vector with $m$-th element
\begin{equation}
  [\avec(\theta)]_m = e^{-j\pi(m-1)\sin\theta}, \quad m = 1,\ldots,M.
  \label{eq:sv}
\end{equation}
The steering vector $\avec(\theta)$ encodes the full geometry of the array response as a function of angle. Specifically, angular proximity implies steering-vector proximity, and vice versa. This key property is utilized in the development of our physics-informed regularizer. 

The received signal covariance, $\mathbf{R} = \sigma_s^2\avec(\theta)\avec^H(\theta) + \sigma^2\mathbf{I}_M$, is estimated by the sample covariance $\hat{\mathbf{R}} = \frac{1}{N}\mathbf{Z}\mathbf{Z}^H$, where $\textbf{Z} = [\mathbf{z}(1), \mathbf{z}(2), \cdots, \mathbf{z}(N)] \in \mathbb{C}^{M \times N}$.


\vspace{-0.3cm}
\subsection{Federated Learning Model}

\textbf{Federated objective.}
Let $K$ clients each hold local dataset $\mathcal{D}_k$ with $n_k = |\mathcal{D}_k|$ samples. The global objective is
\begin{equation}
  \min_{\w}\; F(\w) = \sum_{k=1}^{K} p_k F_k(\w),
  \quad p_k = \frac{n_k}{\sum_j n_j},
  \label{eq:global_obj}
\end{equation}
where $F_k(\w)$ is the local empirical loss with respect to model parameters $\w \in \R^d$. 

\textbf{Physics-informed local objective (\feddoa).}
Standard FL baselines minimize a cross-entropy loss $f_k(\w)$ that assigns equal cost to all angular misclassifications, ignoring the physical structure of the ULA manifold. \feddoa, on the other hand, accounts for the steering-vector regularizer resulting in
\begin{equation}
  \Lcal_k^{\mathrm{D}}(\w) = f_k(\w)
    + \frac{\beta}{2}\norm{\avec(\thetahat_k) - \avec(\theta_k)}^2,
  \label{eq:feddoa_loss}
\end{equation}
where $\thetahat_k$ is the model's predicted DoA, $\theta_k$ is the true DoA at client $k$, and $\beta > 0$ is a weighting hyperparameter. 
The CNN produces a $|\Theta|$-dimensional softmax probability vector $p(\w;\mathbf{X}_k)$, where $|\Theta|$ denotes the number of candidate angles in the set of candidate angles $\Theta$. For training, we obtain a differentiable DoA estimate using the soft-argmax mapping $\thetahat_k = \sum_{i=1}^{|\Theta|} p_i(\w;\mathbf{X}_k)\theta_i$. The regularizer measures discrepancy in the steering space rather than in the label space as traditional FL approaches do. Specifically, because $\|\avec(\theta_1)-\avec(\theta_2)\|$ grows with angular separation, this term imposes a geometry-consistent penalty that is larger for predictions that are farther from the true angle. Domain knowledge about the array manifold is thus embedded in the training objective at each gradient step. Furthermore, note that all norms on $\avec(\theta)\in\C^M$ use the identification $\C^M\cong\R^{2M}$ (to be compatible with real-valued FL models), preserving all metric and gradient properties. We emphasize that $\theta_k$ is a training-time label only, used exactly as the baseline cross-entropy loss already uses it. At inference, $\thetahat_k$ is purely estimated, with no access to $\theta_k$.

\textbf{\doaProx.} \doaProx combines the physics-informed steering regularizer of \feddoa with the proximal term of FedProx \cite{li2020federated}, which constrains local weights to remain close to the last broadcast global model $\wref$ in an effort to prevent client drift from severe levels of heterogeneity in the FL network. The objective for \doaProx is given by
\begin{equation}
  \Lcal_k^{\mathrm{DP}}(\w) = f_k(\w)
    + \frac{\mu}{2}\norm{\w - \wref}^2
    + \frac{\beta}{2}\norm{\avec(\thetahat_k) - \avec(\theta_k)}^2,
  \label{eq:doaProx_loss}
\end{equation}
where $\mu > 0$ is the proximal coefficient. The two regularization terms operate in complementary domains. Specifically, the proximal term suppresses client drift in weight space, while the steering term enforces geometry-consistent angular correction on the array manifold. As we show in Corollary~\ref{cor:msaeP}, this objective yields a strictly tighter MSAE bound than \feddoa alone under heterogeneous data.

\vspace{-0.3cm}
\subsection{Convergence Analysis}


We now analyze the convergence of both methods to a first-order stationary point without any convexity assumption. In our analysis, we employ the following common assumptions. 

\begin{assumption}[$L$-Smoothness]
Each local loss $f_k$ is $L$-smooth:
$\norm{\grad f_k(\w)-\grad f_k(\mathbf{v})}\leq L\norm{\w-\mathbf{v}}$
for all $\w,\mathbf{v}\in\R^d$. This is a common assumption in FL analysis that neural network loss functions generally follow in practice, allowing for the analysis of stable learning rates.
\end{assumption}

\begin{assumption}[Steering Vector Regularity]
The map $a : [\theta_{\min}, \theta_{\max}] \to \mathbb{R}^{2M}$ is $C^2$ and satisfies the metric equivalence: $\exists\, 0 < c_a \le C_a < \infty$ such that for all $\theta_1, \theta_2$,
\begin{equation}
  c_a\abs{\theta_1-\theta_2} \leq\norm{\avec(\theta_1)-\avec(\theta_2)} \leq C_a\abs{\theta_1-\theta_2}.
  \label{eq:metric}
\end{equation}
Using the mean value theorem, we see that $C_a=\mathcal{O}(M^{3/2})$ and, separately, that $c_a=\Omega(\sqrt{M})$  (derivation in Appendix~\ref{app:a2a4}) for narrowband signals. This formalizes that angular and steering-vector proximity are equivalent up to constants, translating the regularizer in~\eqref{eq:feddoa_loss} into an angle-domain guarantee.
\end{assumption}

\begin{assumption}[Bounded Gradient Dissimilarity]
There exists $\delta\geq0$ characterizing client heterogeneity such
that $\sum_kp_{k}\norm{\grad F_k(\w)-\grad F(\w)}^2\leq\delta^2$
for all $\w$~\cite{li2020federated}. This is a standard FL heterogeneity assumption~\cite{li2020federated,karimireddy2020scaffold}, isolating, in our considered framework, label and channel heterogeneity as the source of client drift under non-iid conditions.
\end{assumption}

\begin{assumption}[Bounded Steering Gradient]
$\|\nabla_w R_k(w)\| \le \beta C_a G_p D_a$ for all $w$, where $D_a := \sup\|a(\theta_1)-a(\theta_2)\| \le 2\sqrt{M}$, $R_k(w) = \tfrac{\beta}{2}\big\|a(\hat{\theta}_k)-a(\theta_k)\big\|^2$, and $G_p = \sup_w \big\|\partial\hat{\theta}/\partial w\big\|$ is bounded by the product of layer-wise spectral norms (derivation in Appendix~\ref{app:a4}). Since $\thetahat$ is obtained through the differentiable soft-argmax mapping defined above rather than a hard argmax, the Jacobian $\partial\thetahat/\partial\w$ exists almost everywhere and can be bounded using the chain rule and the layer-wise Jacobian norms. This assumption bounds the regularizer's gradient contribution to control $L_\beta$ and holds for real-world neural networks with bounded-norm layers~\cite{bartlett2017spectrally}.
\end{assumption}

\begin{assumption}[Local Gradient Domination]
There exists $\nu>0$ such that $\norm{\grad_\w\Rcal_k(\w)}^2\geq\nu\,\Rcal_k(\w)$ for all $\w$, i.e., $\Rcal_k$ satisfies a Polyak--\L{}ojasiewicz-type condition. This holds when $\w\mapsto\avec(\hat\theta(\w))$ has a Jacobian bounded away from singularity. We apply this assumption in Corollary~\ref{cor:msae} to convert a gradient-norm bound into an MSAE bound.
\end{assumption}

Now, let $L_\beta = L + \beta L_R$ (see Appendix~\ref{app:smooth}) denote the effective smoothness constant of $\mathcal{L}^{D}$. Using this, and letting $T$ be the total number of communication rounds, we present the convergence bound of \feddoa.

\begin{theorem}[Convergence of \feddoa]
\label{thm:feddoa}
Under Assumptions~1--4, with step size $\eta=1/(L_\beta\sqrt{T})$ and $E=1$ local step per round, the \feddoa iterates satisfy:
\begin{equation}
  \frac{1}{T}\sum_{t=0}^{T-1}\E\!\left[\norm{\grad\Lcal^{\mathrm{D}}(\w^t)}^2\right]
  \leq \frac{2L_\beta\Delta_0 + \delta^2/L_\beta}{\sqrt{T}}
  = \mathcal{O}\!\left(\frac{1}{\sqrt{T}}\right),
  \label{eq:thm1}
\end{equation}
where $\Delta_0 = \Lcal^{\mathrm{D}}(\w^0)-(\Lcal^{\mathrm{D}})^*$, and $(\Lcal^{\mathrm{D}})^*$ is the optimal loss.
\end{theorem}
\begin{proof} See Appendix~\ref{app:thm1}. \end{proof}

\begin{remark}
Theorem~\ref{thm:feddoa} bounds the \emph{average} gradient norm, so some round $t^*=\arg\min_t\E\norm{\grad\Lcal^{\mathrm D}(\w^t)}^2$ satisfies the same bound. Assumptions~5 and~2 then convert this into an MSAE bound at $t^*$ (Corollary~\ref{cor:msae}). Note that we define $\mathrm{MSAE}^{\mathrm D}(T)$ at this best iterate rather than at $\w^T$. 
\end{remark}

\begin{corollary}[MSAE Bound -- \feddoa]
\label{cor:msae}
Under Theorem~\ref{thm:feddoa} and Assumption~5, the mean squared angular error (MSAE)
satisfies:
\begin{equation}
  \mathrm{MSAE}^{\mathrm{D}}(T)
  := \E_k\abs{\thetahat_k(\w^{t^*})-\theta_k}^2
  \leq \frac{4L_\beta(\Delta_0 + \delta^2/L_\beta^2)}{\nu\beta c_a^2\sqrt{T}}.
  \label{eq:msae}
\end{equation}
Since the MSAE bound depends on the lower metric-equivalence constant through $1/c_{a}^{2}$, the $c_a = \Omega(\sqrt{M})$ scaling yields an $\mathcal{O}(1/M)$ improvement with array aperture. In contrast, $C_{a} = \mathcal{O}(M^{3/2})$ enters the effective smoothness constant $L_\beta$ and therefore affects the optimization-dependent prefactor rather than directly determining the aperture scaling of the MSAE bound.
\end{corollary}

\begin{proof} See Appendix~\ref{app:cor1}. \end{proof}

For \doaProx, the proximal term raises effective smoothness to $L_{\beta\mu}=L_\beta+\mu$. Although $\mu(\w^t-\wref)$ is common to every client and so cancels from the raw dissimilarity of Assumption~3, the reduction to $\delta_\mu^2$ instead follows from the non-expansiveness of the proximal update, which contracts the spread between clients' iterates (full derivation in Appendix~\ref{app:thm2}). This yields the effective dissimilarity $\delta_\mu^2 \le \delta^2/(1+\eta\mu)^2 \le \delta^2$ used below. We now present the convergence bound of \doaProx. 

\begin{theorem}[Convergence of \doaProx]
\label{thm:doaProx}
Under Assumptions~1--4, with $\eta=1/(L_{\beta\mu}\sqrt{T})$:
\begin{equation}
  \frac{1}{T}\!\sum_{t=0}^{T-1}\E\!\left[\norm{\grad\Lcal^{\mathrm{DP}}(\w^t)}^2\right]
  \leq \frac{2L_{\beta\mu}\Delta_0^{\mathrm{DP}} + \delta_\mu^2/L_{\beta\mu}}{\sqrt{T}},
  \label{eq:thm2}
\end{equation}
resulting in $\mathcal{O}(1/\sqrt{T})$, where $\Delta_0^{\mathrm{DP}}=\Lcal^{\mathrm{DP}}(\w^0)-(\Lcal^{\mathrm{DP}})^*$.
\end{theorem}

\begin{proof} See Appendix~\ref{app:thm2}. \end{proof}

\begin{corollary}[MSAE Bound -- \doaProx]
\label{cor:msaeP}
Under Theorem~\ref{thm:doaProx} and Assumption~5 (evaluated, as in Corollary~\ref{cor:msae}, at the best iterate $t^*$):
\begin{equation}
  \mathrm{MSAE}^{\mathrm{DP}}(T)
  \leq \frac{4L_{\beta\mu}(\Delta_0^{\mathrm{DP}}+\delta_\mu^2/L_{\beta\mu}^2)}
             {\nu\beta c_a^2\sqrt{T}}
  \;\leq\; \mathrm{MSAE}^{\mathrm{D}}(T).
  \label{eq:msaeP}
\end{equation}
The inequality $\mathrm{MSAE}^{\mathrm{DP}}\leq\mathrm{MSAE}^{\mathrm{D}}$ follows from $\delta_\mu^2\leq\delta^2$: the proximal term reduces angular estimation error under data heterogeneity by suppressing the client drift that would otherwise corrupt the physics-informed gradient signal.
\end{corollary}

\begin{proof} See Appendix~\ref{app:cor2}. \end{proof}



\vspace{-0.4cm}
\section{Experimental Setup and Results}

\subsection{Experimental Setup} \label{exp_setup}

\begin{table}[t]
\centering
\footnotesize
\renewcommand{\arraystretch}{0.95}
\caption{CNN Classifier Architecture~\cite{wang2025speculative}}
\label{tab:arch}
\begin{tabular}{lcc}
\toprule
\textbf{Layer} & \textbf{Activation} & \textbf{Shape} \\
\midrule
Conv 1                 & ReLU    & $2\times2\times32$ \\
BatchNorm 1            & ---     & $2\times2\times32$ \\
MaxPool 1              & ---     & $1\times2\times32$ \\
Conv 2                 & ReLU    & $2\times2\times64$ \\
BatchNorm 2            & ---     & $2\times2\times64$ \\
MaxPool 2              & ---     & $1\times2\times64$ \\
Flatten                & ---     & --- \\
Dense 1 (Dropout 30\%) & ReLU    & 128 \\
Output                 & Softmax & 61 \\
\bottomrule
\end{tabular}
\end{table}

We generate 180,000 signals in MATLAB, each with $N=1500$ snapshots using $M=8$ array elements. Each signal, $\textbf{Z} = [\mathbf{Z}_1, \mathbf{Z}_{2}] \in \mathbb{C}^{M \times N}$, is comprised of $\mathbf{Z}_1 \in \mathbb{C}^{M \times N/2}$, which contains only noise, and $\mathbf{Z}_2 \in \mathbb{C}^{M \times N/2}$, which contains both the signal of interest and noise. We concatenate the two covariance matrices and separate the real and imaginary components to obtain input $\mathbf{X} = [\mathbf{R}_1, \mathbf{R}_2] \in \mathbb{R}^{M \times 2M \times 2}$. The angle of arrival was restricted to $[-60^\circ,60^\circ]$ in $2^\circ$ increments, for 61 candidate angles. 


We evaluated four experimental conditions: (1) label-iid/channel-iid, (2) label-iid/channel-non-iid, (3) label-non-iid/channel-iid, and (4) label-non-iid/channel-non-iid. The number of clients was $K\in\{6,12\}$. 
We modeled channel fading using Gaussian, Rayleigh, and Rician distributions. For label non-iid settings, angle sorted samples were evenly partitioned among clients; with $K=12$, eleven clients received five adjacent angles and one received six. For channel non-iid settings, clients were evenly divided among Gaussian, Rayleigh, and Rician channels, with each client observing only one distribution. For the combined non-iid setting, the three channel distributions were restricted to $[-60^\circ, -20^\circ]$, $[-18^\circ, 20^\circ]$, and $[22^\circ, 60^\circ]$, respectively, and the angle partitioning was applied separately within each channel. 
We instantiate the CNN of~\cite{wang2025speculative}, reproduced in Table~\ref{tab:arch}, for our classifier. 



We compare \feddoa and \doaProx to three non-physics-informed state-of-the-art FL baselines: FedAvg~\cite{mcmahan2017communication}, FedProx~\cite{li2020federated}, and SCAFFOLD~\cite{karimireddy2020scaffold} as well as a centralized learning scheme for reference. The chosen FL baselines were selected to isolate the steering term from the proximal term. Specifically, \feddoa uses just the steering term, FedProx uses only the proximal term, and \doaProx uses both the steering and proximal terms. In all simulations, we used $T = 200$ (except the centralized scheme, which used $T=50$) and minibatch sizes of 128. SCAFFOLD used $\eta = 3\cdot10^{-4}$ for label-non-iid conditions and $\eta = 10^{-4}$ otherwise; all other algorithms used $1\cdot10^{-3}$. Physics-informed weights were $\beta=10^{-3}$ and $\mu=10$. All hyperparameters were selected via hyperparameter search per algorithm. 

Note that the regularizer adds only $\mathcal{O}(M)$ element-wise cost per sample, negligible next to the CNN's convolutional FLOPs (Table~\ref{tab:arch}), with no additional forward passes. Since it changes neither the architecture nor the transmitted weight dimension $d$, \feddoa and \doaProx add no communication overhead relative to FedAvg, FedProx, or SCAFFOLD.



\begin{figure}[t]
\centering
\includegraphics[width=0.35\textwidth]{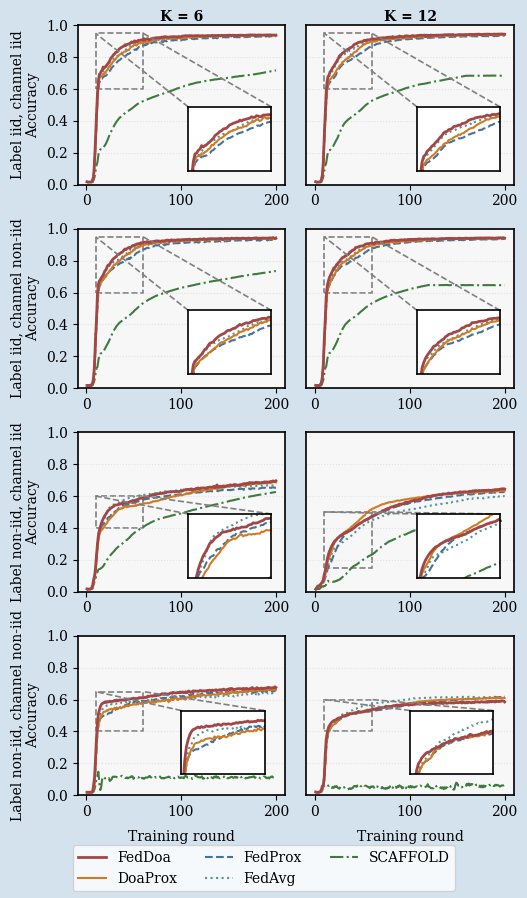}
\caption{Training curves over validation set for all experimental conditions and client counts. Physics-informed methods (\feddoa and \doaProx) consistently reach higher final accuracy and converge more rapidly than each baseline.} \label{fig1}
\end{figure}
\begin{figure}[t]
\centering
\includegraphics[width=0.35\textwidth]{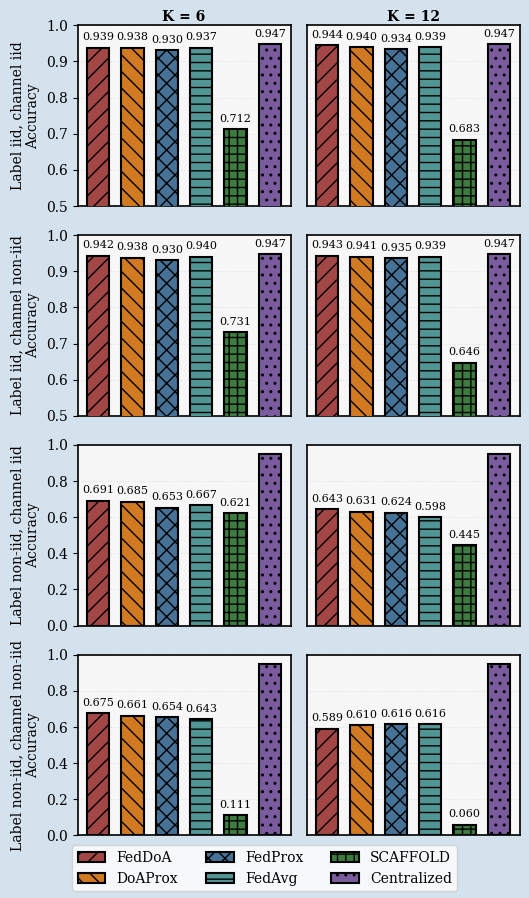}
\caption{Average validation accuracy over last 10 rounds for each experimental condition. \feddoa achieves the highest FL accuracy in all but one condition and demonstrates the smallest degradation from iid to non-iid settings.} \label{fig2}
\end{figure}

\begin{table}[t]
\footnotesize
\centering
\caption{MSAE for $K=6$.}
\label{tab:msae_k6}
\setlength{\tabcolsep}{4pt}
\begin{tabular}{lcccc}
\toprule
\textbf{Method} & \textbf{LI/CI} & \textbf{LN/CI} & \textbf{LI/CN} & \textbf{LN/CN} \\
\midrule
\feddoa   & 56.795  & 296.421 & 53.245  & 317.462 \\
\doaProx  & 61.943  & 302.247 & 57.853  & 325.620 \\
FedAvg \cite{mcmahan2017communication}   & 57.783  & 323.016 & 54.869  & 346.246 \\
FedProx \cite{li2020federated}  & 67.412  & 337.788 & 64.126  & 338.119 \\
SCAFFOLD \cite{karimireddy2020scaffold} & 267.232 & 340.337 & 258.440 & 491.503 \\

\bottomrule
\end{tabular}

\vspace{2pt}
\footnotesize{L = label, C = channel, I = iid, N = non-iid.}
\end{table}



\begin{table}[t]
\footnotesize
\centering
\caption{Hyperparameter sweep for K = 6, label non-iid, channel iid.}
\label{tab:hp}
\setlength{\tabcolsep}{4pt}
\begin{tabular}{lccc}
\toprule
\textbf{Algorithm} & \textbf{$\mu = 1$} & \textbf{$\mu = 10$} & \textbf{$\mu = 100$} \\
\midrule
FedProx   & 0.697  & 0.653 & 0.657 \\
\midrule
\textbf{Algorithm} & \textbf{$\beta = 10^{-4}$} & \textbf{$\beta = 10^{-3}$} & \textbf{$\beta = 10^{-2}$} \\
\midrule
\feddoa & 0.670 & 0.675 & 0.670 \\

\bottomrule
\end{tabular}
\end{table}

\vspace{-0.5cm}
\subsection{Performance Evaluation} \label{perf_eval}

Fig. \ref{fig1} and Fig. \ref{fig2} consistently demonstrate the advantage of physics-informed training. Under each considered condition, all methods achieve convergence, but \feddoa and \doaProx converge faster, typically within the first 50--80 rounds, whereas FedAvg and FedProx converge slower, plateauing at slightly lower values. The steering regularizer accelerates convergence by penalizing geometrically large angular errors from initial training, before the cross-entropy loss alone has shaped the decision boundary. Moreover, Table ~\ref{tab:msae_k6} shows the final MSAE for each experiment when $K=6$. Here, we see that \feddoa produces the lowest MSAE. Identical trends were observed for $K=12$. Moreover, FedDoA and DoAProx achieve performance comparable to centralized training under iid conditions, while exhibiting a larger performance gap under non-iid conditions due to client heterogeneity.

Finally, to demonstrate the sensitivity of our approach to hyperparameters, we show the performance of our framework in a challenging setting ($K=6$, label-non-iid, channel-iid) under various values of $\beta$ and $\mu$  in Table \ref{tab:hp}. Note that Table \ref{tab:hp} reports results on final validation accuracy whereas Fig. \ref{fig2} averages performance over the last 10 rounds. The ablation results indicate that FedDoA is relatively insensitive to $\beta$ over the tested range, whereas FedProx exhibits greater sensitivity to $\mu$, motivating the use of moderate regularization coefficients in the main experiments.

\vspace{-0.4cm}
\section{Conclusion}

We introduced \feddoa and \doaProx, the first physics-informed federated learning algorithms for DoA estimation, which embed steering-vector geometry directly into the local training objective of federated learning. Future work will investigate multi-source generalizations as well as array imperfections such as nonlinear distortions.


\vspace{-0.4cm}

\bibliographystyle{ieeetr}
\bibliography{bibliography}

\vspace{-0.7cm}
\appendices

\section{Proof of Assumptions~2 and~4}
\label{app:a2a4}
\label{app:a4}

\begin{proof}[Assumption~2 scaling]
Differentiating \eqref{eq:sv}, we obtain $[a'(\theta)]_m = -j\pi(m-1)\cos\theta e^{-j\pi(m-1)\sin\theta}$, and therefore $|a'(\theta)|_2 = \pi|\cos\theta| \sqrt{\frac{M(M-1)(2M-1)}{6}} = \mathcal{O}(M^{3/2})$. By the mean-value theorem, this gives the upper Lipschitz bound $|a(\theta_1)-a(\theta_2)|_2 \leq C_a|\theta_1-\theta_2|$, with $C_a = \pi \sqrt{\frac{M(M-1)(2M-1)}{6}} = \mathcal{O}(M^{3/2})$. For $c_a$, with $u = \sin\theta$, $\|a(\theta_1)-a(\theta_2)\|^2 = 2M - 2\sum_{m=0}^{M-1}\cos(\pi m \Delta u)$. Since $d=\lambda/2$ and $|\theta|\le 60^\circ$ give $|u_1-u_2| \le \sqrt{3} < 2$, no grating lobe arises, so the sum is $\mathcal{O}(1)$ at maximal separation, where the ratio is minimized, giving $c_a = \Omega(\sqrt{M})$.
\end{proof}

\begin{proof}[Assumption~4]
By the differentiability of the soft-argmax mapping defining $\thetahat$, $\grad_\w\Rcal_k=\beta(\partial\thetahat/\partial\w)^\top\avec'(\thetahat)^\top(\avec(\thetahat)-\avec(\theta_k))$. Taking norms and applying Assumption~2: $\|\nabla_w R_k\| \le \beta\big\|\partial\hat{\theta}/\partial w\big\|C_a \|a(\hat{\theta})-a(\theta_k)\| \le \beta C_a G_p D_a$, with $G_p = \sup_w \big\|\partial\hat{\theta}/\partial w\big\|$ bounded by the product of layer-wise spectral norms, independent of the trajectory, and $D_a \le 2\sqrt{M}$ the manifold diameter.
\end{proof}

\vspace{-0.5cm}
\section{Proof of Theorem~\ref{thm:feddoa} and Corollary~\ref{cor:msae}}
\label{app:smooth}
\label{app:thm1}
\label{app:cor1}

\begin{proof}[Effective smoothness]
Differentiating $R_k$ twice and applying Assumption~2 with $\|a''\| \le H_a$ and $\|\partial^2\hat{\theta}/\partial w^2\| \le H_p$ gives $\|\mathrm{Hess}_w R_k\| \le \beta\big(C_a^2 G_p^2 + D_a(H_a G_p^2 + C_a H_p)\big) =: \beta L_R$; with $L$-smooth $f_k$, $\mathcal{L}^{D}$ is $L_\beta$-smooth, $L_\beta = L + \beta L_R$.
\end{proof}

\begin{proof}[Theorem~\ref{thm:feddoa}]
$L_\beta$-smoothness gives
\begin{align}
  \Lcal^{\mathrm{D}}(\w^{t+1}) &\leq \Lcal^{\mathrm{D}}(\w^t)
  -\eta\inner{\grad\Lcal^{\mathrm{D}}(\w^t)}{\g^t}
  +\tfrac{\eta^2L_\beta}{2}\norm{\g^t}^2. \tag{i}
\end{align}
Taking expectation with $\E[\g^t|\w^t]=\grad\Lcal^{\mathrm{D}}(\w^t)$
and bounding the second moment via Assumption~3:
\begin{equation}
  \E\norm{\g^t}^2\leq\norm{\grad\Lcal^{\mathrm{D}}(\w^t)}^2+\delta^2. \tag{ii}
\end{equation}
Substituting and using $\eta\leq1/L_\beta$ so $1-\eta L_\beta/2\geq1/2$:
\begin{equation}
  \E[\Lcal^{\mathrm{D}}(\w^{t+1})]\leq\Lcal^{\mathrm{D}}(\w^t)
  -\tfrac{\eta}{2}\norm{\grad\Lcal^{\mathrm{D}}(\w^t)}^2
  +\tfrac{\eta^2L_\beta}{2}\delta^2. \tag{iii}
\end{equation}
Summing $t=0,\ldots,T-1$, telescoping, and using
$\E[\Lcal^{\mathrm{D}}(\w^T)]\geq(\Lcal^{\mathrm{D}})^*$ gives
\begin{equation}
  \tfrac{\eta}{2}\textstyle\sum_t\E\norm{\grad\Lcal^{\mathrm{D}}(\w^t)}^2
  \leq\Delta_0+\tfrac{T\eta^2L_\beta}{2}\delta^2. \tag{iv}
\end{equation}
Dividing by $\eta T/2$ and setting $\eta=1/(L_\beta\sqrt{T})$ yields Theorem~\ref{thm:feddoa}.
\end{proof}

\begin{proof}[Corollary~\ref{cor:msae}]
Let $t^*=\arg\min_{t\in\{0,\ldots,T-1\}}\E\norm{\grad\Lcal^{\mathrm D}(\w^t)}^2$; since the minimum is no larger than the average, (iv) gives $\E\norm{\grad\Lcal^{\mathrm D}(\w^{t^*})}^2 \leq \frac{2L_\beta\Delta_0+\delta^2/L_\beta}{\sqrt T}$. Since $\Rcal_k$ is a smooth summand of $\Lcal_k^{\mathrm D}$, $\norm{\grad_\w\Rcal_k(\w^{t^*})}\leq\norm{\grad\Lcal^{\mathrm D}(\w^{t^*})}+\norm{\grad f_k(\w^{t^*})}$; averaging over $k$ and applying Assumption~5 gives $\E_k[\Rcal_k(\w^{t^*})]\leq\frac{1}{\nu}\E_k\norm{\grad_\w\Rcal_k(\w^{t^*})}^2$, bounded (up to constants absorbed into $\nu$) by the above. By Assumption~2, $\Rcal_k(\w)\geq\frac{\beta c_a^2}{2}|\thetahat_k-\theta_k|^2$; averaging over clients at $\w^{t^*}$ and combining yields Corollary~\ref{cor:msae}.
\end{proof}

\vspace{-0.5cm}

\section{Proof of Theorem~\ref{thm:doaProx} and Corollary~\ref{cor:msaeP}}
\label{app:thm2}
\label{app:cor2}

\begin{proof}[Theorem~\ref{thm:doaProx}]
Steps (i)--(iv) above follow with $L_\beta\to L_{\beta\mu}=L_\beta+\mu$. Since $\mu(\w^t-\wref)$ is identical across clients, it cancels from the raw dissimilarity of Assumption~3. The reduction to $\delta_\mu^2$ instead follows from non-expansiveness of the proximal update, $\norm{\w_j^{t+1}-\w_k^{t+1}}\leq\frac{1}{1+\eta\mu}\norm{\eta\grad f_j(\w^t)-\eta\grad f_k(\w^t)}$, which contracts client spread by $1/(1+\eta\mu)$; combined with $L$-smoothness (Assumption~1), $\delta_\mu^2 \le \delta^2/(1+\eta\mu)^2 \le \delta^2$, where $\delta_\mu^2$ is the dissimilarity at the clients' own iterates rather than a reduction in $\delta$ itself. Thus,
\begin{equation}
\mathbb{E}\big\|g^t\big\|^2 \;\le\;
\big\|\nabla \mathcal{L}^{DP}(w^t)\big\|^2 + \delta_\mu^2 .
\tag{ii$'$}
\end{equation}
Repeating steps (iii)--(iv) with $\delta^2\to\delta_\mu^2$ yields Theorem~\ref{thm:doaProx}.
\end{proof}

\begin{proof}[Corollary~\ref{cor:msaeP}]
Identical to the proof of Corollary~\ref{cor:msae} with $L_\beta\to L_{\beta\mu}$ and $\delta^2\to\delta_\mu^2$; $\mathrm{MSAE}^{\mathrm{DP}}\leq\mathrm{MSAE}^{\mathrm{D}}$ follows from $\delta_\mu^2\leq\delta^2$.
\end{proof}

\end{document}